\documentclass{article}
\usepackage{graphicx,amssymb,amsmath,amsthm}
\usepackage{cite}
\usepackage{color}
\usepackage{a4wide}
\newtheorem{theorem}{Theorem}
\newtheorem{cor}{Corollary}
\newtheorem{obs}{Observation}
\newtheorem{lemma}{Lemma}

\title{Cannibal Animal Games: \\  a new variant of Tic-Tac-Toe\footnote{ Work partially supported by the ESF EUROCORES programme EuroGIGA, CRP ComPoSe: grant EUI-EURC-2011-4306, for Spain.}}

\author{
Jean Cardinal\thanks{Department of Computer Science, Universit\'e Libre de Bruxelles (ULB), Belgium, {\tt \{jcardin, sebastien.collette, stefan.langerman, perouz.taslakian\}@ulb.ac.be}}
\and
S\'{e}bastien Collette${}^*$
\and
Hiro Ito\thanks{School of Informatics and Engineering, The University of Electro-Communications, Japan, 
{\tt itohiro@uec.ac.jp}}
\and
Matias Korman\thanks{Universitat Polit\`ecnica de Catalunya (UPC), Barcelona. {\tt{matias.korman@upc.edu}}. With the support of the Secretary for Universities and Research of the Ministry of Economy and Knowledge of the Government of Catalonia and the European Union.}
\and
Stefan Langerman${}^*$\thanks{Ma\^itre de Recherches du FRS-FNRS.} 
\and
Hikaru Sakaidani\thanks{
School of Informatics, Kyoto University, Japan, 
{\tt sakaidani@lab2.kuis.kyoto-u.ac.jp}}
\and
Perouz Taslakian${}^*$ 
}

\date{}

\begin{document}
\maketitle

\begin{abstract}
This paper presents a new partial two-player
game, called the \emph{cannibal animal game},  which is a  variant of Tic-Tac-Toe. 
The game is played on the infinite grid, where in each round a player chooses and occupies free cells. 
The first player Alice can occupy a cell in each turn and 
wins if she occupies a set of cells, the union of a subset of which is a 
translated, reflected and/or rotated copy of a previously agreed upon polyomino $P$ (called an \emph{animal}).
The objective of the second player Bob is to 
prevent Alice from creating her animal by occupying in each round
a translated, reflected and/or rotated copy of $P$.
An animal is a \emph{cannibal} if Bob has a winning strategy, 
and a \emph{non-cannibal} otherwise. 
This paper presents some new tools, such as the \emph{bounding strategy} and the \emph{punching lemma},
to classify animals into cannibals or non-cannibals.  
We also show that the \emph{pairing strategy} works for this problem. 
\end{abstract}

\section{Introduction}

Variants of the Tic-Tac-Toe game have been the focus of a number of studies in the area of 
recreational mathematics~\cite{winningways3,1-2-achievementgame,TTT_Gardner,Eureka42,avoid,
handicap_square,ItoHERCMA,tttsurvey}. 
Probably the most studied among these games is an achievement game,
a class of generalized Tic-Tac-Toe games presented by Harary~\cite{TTT_Gardner,Eureka42}.
A {\em polyomino} or an {\em animal} is a set of connected cells (in the 4-neighbor topology) 
of the infinite grid. In the Harary games~\cite{TTT_Gardner} defined by a given animal, two players Alice and Bob alternatively occupy one cell 
in each round of the game (we assume that Alice is the first player), and the first player who occupies a translated copy of the given animal is the winner. By the strategy stealing argument, Bob 
cannot win in these type of games. Thus, his objective is to obstruct Alice's achievement. 

Here we present a new achievement game called  
the {\em cannibal animal game}. 
As with Harary's game, it is played on the infinite grid
whereby players alternate turns to occupy free cells of the grid. 
This means that in each round the player must choose grid cells that are not yet occupied. Once a cell is occupied, it remains so until the end of the game. %
 In contrast to the generalized Tic-Tac-Toe, the cannibal animal game is a 
\emph{partial game}: 
the roles and legal moves of Alice and Bob are different. 
Alice's legal move is to occupy one cell of the infinite grid in each round, and 
she wins if she occupies a translated copy of an animal given beforehand 
(this move is the same as that of the first player of Harary's generalized Tic-Tac-Toe). 
Bob's role and allowed moves, however, are different: 
in each round he must occupy a copy of the given animal (i.e., occupy a subset of the grid cells),
and his objective is to prevent Alice from achieving the animal. Neither Alice nor Bob's moves are allowed to overlap with already occupied regions, even partially. The animal achieved or that Bob occupies may be a translation, a mirror image and/or a $90$, $180$, or $270$-degree rotation of the given animal.
Each such translation/rotation/reflection is called a \emph{copy} of the animal. %
 Figure~\ref{fig_el} shows an example of the progress of the game where the animal is \emph{El}, an L-shaped triomino.

%
\begin{figure}[h]
\center
\includegraphics[width=0.4\textwidth]{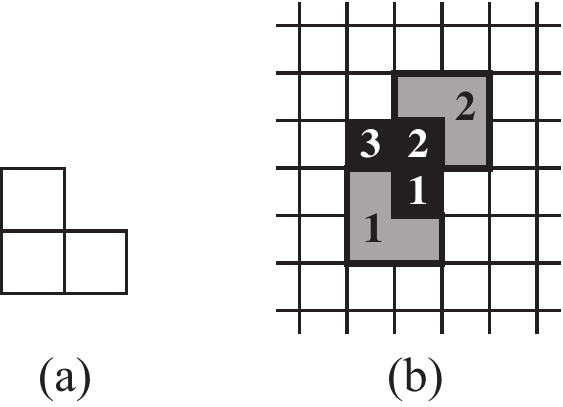}
\caption{(a) The animal El (an L-shaped triomino), (b) An example of the progress of the game: 
cells depicted in black are occupied by Alice, and animals depicted in gray correspond to Bob's moves. 
In both cases, the numbers on the cells represent the order in which the cells are occupied. 
In the example, Alice wins.}
\label{fig_el}
\end{figure}

Any animal of $n$ cells is called 
an {$n$-cell-animal} (alternative, we refer $n$ as the {\em size} of the animal).  Also, let $[x_{\min},x_{\max}]\times [y_{\min},y_{\max}]$ be the rectangular region defined by the corner cells $(x_{\min},y_{\min})$, $(x_{\max},y_{\min})$, $(x_{\min},y_{\max})$, and $(x_{\max},y_{\max})$. We call an animal a 
{\em cannibal} or a \emph{loser} if Bob has a winning strategy 
(Bob's animal eats Alice's animal) and 
a {\em non-cannibal} or a \emph{winner} otherwise. 
And hence the game is called the \emph{cannibal animal game}. The region in which Alice and Bob place their pieces will be called {\em board} and {\em grid} indistinctively.

\paragraph{Our Results.}
In this paper we study the following animals 
(see Figure~\ref{fig_shapes} for examples): 
$R(n,m)$ 
is an $n \times m$ rectangle. 
We also define $O(n,m,k)$ (for $n,m\in\mathbb{N}$ and $k < \min \{ n/2, m/2 \}$) as 
a $2k(n+m-2k)$-cell-animal having the shape of $R(n,m)$
but with a $(n-2k) \times (m-2k)$ rectangular hole in the center 
(that is, an an O-shaped animal of thickness $k$). 
Animal $U(h,w,k)$ (for $h\geq 2$, $w \geq 3$, and $k < \min \{ h, w/2 \}$)
is defined as a $k(2h+w-2k)$-cell-animal having a $U$-shape with height $h$, width $w$, and thickness $k$. 
The $L(n)$ animal (for any $n\in \mathbb{N}$) consists of the concatenation of $n$ copies of the El-animal, translated horizontally so that they touch, but do not overlap. 

%
\begin{figure}
\center
\includegraphics[width=0.6\textwidth]{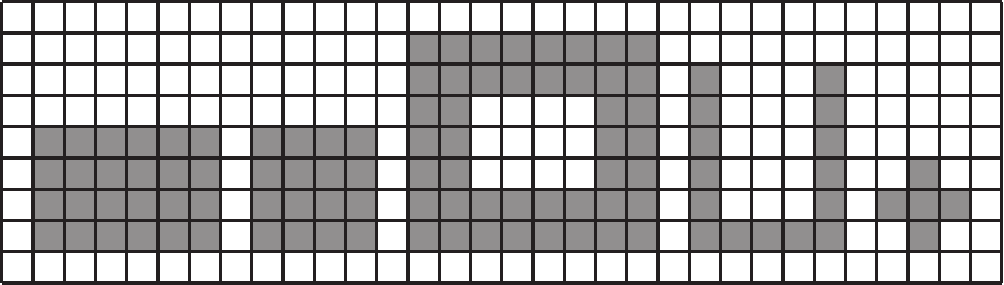}
\caption{Examples of animals: 
$R(4,6)$, $R(4,4)$, $O(7,8,2)$, $U(6,5,1)$, and $L(2)$ (from left to right).}
\label{fig_shapes}
\end{figure}

\begin{enumerate}
\item The following animals are cannibals:
\begin{enumerate}
\item $R(n,n)$ with holes if at least one of the holes is at least 
$\lfloor n/4 \rfloor$ cells away from the boundary for $n \geq 4$ 
(and no hole is on the boundary)
\item $O(n,m,k)$ for $n,m\in \mathbb{N}$, and $k < \min \{ n/2, m/2 \}$ 
\item $U(h,w,1)$ for $h\geq 2$, $w \geq 3$, except $U(2,4,1)$
\item $L(n)$ for $n \geq 2$
\end{enumerate}
\item The following animals are non-cannibals:
\begin{enumerate}
\item Animals with at most three cells 
\item $R(n,m)$ for any $n,m \in \mathbb{N}$
\end{enumerate}
\end{enumerate}



\section{Non-cannibal animals (winners) and the bounding strategy}\label{sec_win}
In this section we present some non-cannibal animals. First we start by observing that when an animal is small, Alice can easily win.

\begin{obs}\label{obs_three} 
Any animal $P$ of three or fewer cells is non-cannibal.
\end{obs} 
We also conjecture that polyominoes of size $4$ are all non-cannibal, but we have not been able to construct a winning strategy for all of them. 
%
%
 In the following we give winning strategies for Alice for the case in which the polyomino is a rectangle.
%

 
\begin{theorem}\label{theo_rect} 
$R(n,m)$ is a non-cannibal (for any $n,m\in \mathbb{N}$).
\end{theorem} 
%
To prove the theorem, we first give a strategy for the case in which the board is bounded. This will afterwards be used for the unbounded board.

\begin{lemma}\label{lem_rect} 
In any finite board, the rectangle $R(n,m)$ is non-cannibal provided that at least one copy of $R(n,m)$ can be placed on the empty board.
\end{lemma} 
\begin{proof}
At the beginning of each round we define $\mathcal{S}=\{s_1,\ldots, s_k\}$ as the set of copies of $P$ 
not occupied by Bob that fit on the board (note that some of these positions may be occupied by Alice's previous moves).  
The set $\mathcal{S}$ will be treated as a set of potential positions in which Alice may form her animal. 
Note that Bob's moves must be at some $s \in \mathcal{S}$. Also, let $\mathcal{S'}\subseteq \mathcal{S}$ be the set of animals that stab all elements of $\mathcal{S}$ (that is, $s'\in\mathcal{S'} \Leftrightarrow s'\cap s\neq \emptyset$, $\forall s \in\mathcal{S}$). Note that the set $\mathcal{S}$ initially is nonempty at the beginning of the game, and whenever Bob plays, the size of $S$ is reduced. Moreover, the set $\mathcal{S}$ will only become empty if and only if Bob manages to place his copy occupying the cells of some $s'\in \mathcal{S'}$. 

The key observation is the fact that $\mathcal{S'}$ is a collection of pairwise intersecting rectangles, and as such it must have at least a common intersection point $c_\mathcal{S'}$ that intersects all rectangles of $\mathcal{S'}$. Alice's strategy is as follows: if the set $\mathcal{S'}$ is empty, Alice occupies any empty cell of some $s\in \mathcal{S}$. Otherwise, Alice plays at $c_\mathcal{S'}$, preventing Bob from playing at $\mathcal{S'}$. 

With this strategy, Alice makes sure that the set $\mathcal{S}$ never becomes empty 
(since Bob can never occupy $s'\in \mathcal{S'}$). Since the number of Bob's possible moves only decreases after each of Alice's moves, after a finite number of turns Bob will be unable to play inside the bounded board (and Alice will be able to complete a copy of the animal). 
\end{proof}

Observe that the proof of Lemma~\ref{lem_rect} makes no assumptions on the shape of the finite board (other that a copy of $R(n,m)$ fits inside, and  that Alice plays first). In the following we extend this result to an infinite board. The first step in  
Alice's strategy will be to construct a bounded region big enough so that the set $\mathcal{S}$ is nonempty, and then apply the bounded region strategy. 
From this idea we have the proof of Theorem~\ref{theo_rect} as follows: 
\\

\noindent\textbf{Proof of Theorem~\ref{theo_rect} (Bounding strategy).}
We construct a region on the board large enough 
that at least one copy of $R(n,m)$ can be constructed 
inside. The objective is to create an $N\times N$ square for a sufficiently large $N$ (the exact value will be determined later). Alice can surround the boundary of the square with at most $4(N-2)$ moves (note that the four corners need not be occupied). 
Let $I$ be the interior of the square. 
Notice that at least $(N-(n-1))(N-(m-1))$ copies of $R(n,m)$  
fit inside $I$. 
Each of Bob's animals stabs at most 
$(2n-1)(2m-1)+(n+m-1)^2 \leq n^2 + m^2 + 6nm$ 
copies of $R(n,m)$. 

During the (at most) $4N$ rounds in which Alice surrounds the boundary of the square, 
Bob can stab
at most $4N(n^2 + m^2 + 6nm)$ animals of $\mathcal{S}$. 
Thus, if $(N-n+1)(N-m+1) > 4N(n^2 + m^2 + 6nm)$, 
the set $\mathcal{S}$ will be non-empty even after Alice has completed 
surrounding the boundary of the square. 
Because the first term is quadratic in $N$ and the second is linear, for a sufficiently large $N$ the inequality holds. 
\hspace*{\fill} $\Box$

The key property of this strategy is the fact that any collection of pairwise intersecting rectangles has a common intersection point. Hence, this approach could be extended to any other animal that also satisfies this property. This property is often referred as the 2-Helly (or simply the Helly) property~\cite{handbookDCG97} in the literature. Unfortunately, in a companion paper~\cite{helly} we show that the rectangle is the only 2-Helly polyomino. 


Observe that the above strategy might take many moves, since Alice starts by enclosing a large region. In the following, we provide a strategy that uses fewer moves for the particular case in which the animal is an $n\times n$ square. Let $S_n(x,y)$ be the connected square region of the grid 
such that its bottom-left cell is located at position $(x,y)$;
that is, $S_n(x,y)$ occupy the square region defined by the rectangular region $[x,x+n-1]\times [y,y+n-1]$. 

\begin{lemma}
For any $n>0$, Alice can construct $R(n,n)$ using at most $n^2+3$ moves. 
\end{lemma}

\begin{proof}
We will describe Alice's strategy for constructing $R(n,n)$: Alice will try to play at 
locations $(0,0)$, $(n,0)$, $(2n,0)$ and $(3n,0)$, forming a horizontal strip 
with none of Bob's pieces (see Figure~\ref{Sn-case1}). By virtually rotating the board we can certify that Bob's first move will be to the left of Alice's. Hence, she can occupy positions $(0,0)$ and $(n,0)$ with her first two moves. 
Alice's strategy now depends on whether Bob allows Alice to play in the third and fourth positions.
\begin{figure}[h]
\center
\includegraphics[width=0.6\textwidth]{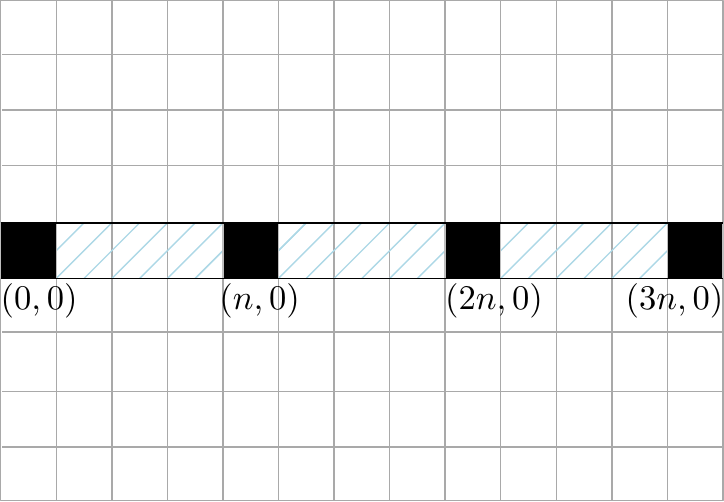}
\caption{Alice's strategy is to occupy locations $(0,0)$, $(n,0)$, $(2n,0)$, and $(3n,0)$.}
\label{Sn-case1}
\end{figure}

\begin{description}

\item[Two cells occupied] The only case in which Alice cannot occupy position $(2n,0)$ in her third move is if Bob plays his second move occupying position $(2n,0)$. 
We claim that in such a case, Alice will win in at most $n^2 + 2$ moves.
To prove the claim, let $S_n(b_x,b_y)$ be the position in which Bob placed his second move. Observe that we must have 
 $n+1 \leq b_x \leq 2n$ and $-n+1 \leq b_y \leq 0$.
Then Alice plays her third move at position $(b_x-n,n-1)$. 
\begin{figure}[h]
\center
\includegraphics[width=0.6\textwidth]{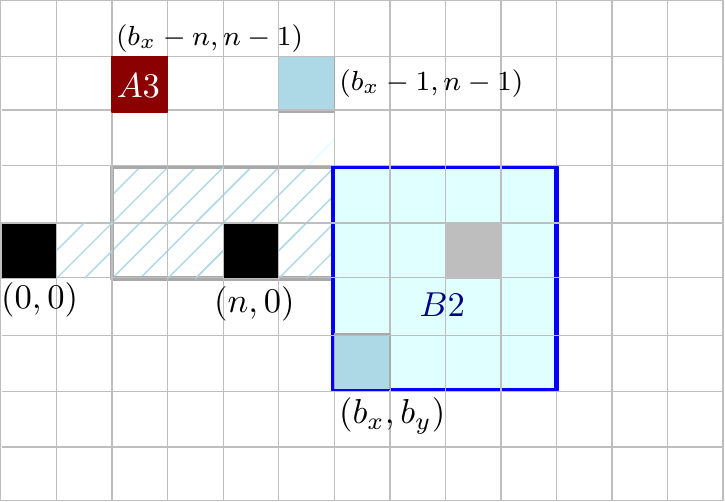}
\caption{If Bob forbids Alice from occupying $(2n,0)$ with his second move $B2$ having its bottom-left cell at $(b_x,b_y)$, 
then Alice plays her third move at $A_3$.
The striped rectangular area is a region that Bob cannot occupy.}
\label{Sn-case2-1}
\end{figure}

At this point, note that the cells inside the rectangular region $[b_x-n,b_x-1]\times[0,b_y+n-1]$ can no longer be occupied by Bob since the (horizontal and vertical) distance between any two of the cells occupied by Alice within this region is less than $n$ (Figure~\ref{Sn-case2-1}). 
If $b_y = 0$, this rectangular region defines 
$S_n(b_x-n,0)$ 
and since no cell of this square region 
can be occupied by Bob, then 
Alice wins by occupying $S_n(b_x-n,0)$ in $n^2 + 1$ moves (Figure~\ref{Sn-case2-1}). 

\begin{figure}[h]
\center
\includegraphics[width=0.6\textwidth]{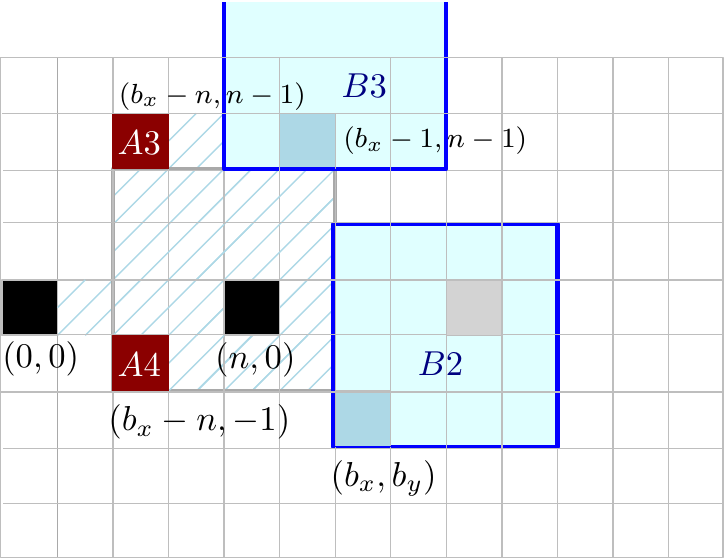}
\caption{If Bob forbids Alice from occupying $S_n(b_x-1,-1)$ with his third move $B3$, 
then Alice plays her fourth move at $A_4$.
The striped rectangular area is a region that Bob cannot occupy.}
\label{Sn-case2-2}
\end{figure}

Otherwise, we have $-(n-1) \leq b_y < 0$. Notice that, if Alice plays at positions $(b_x-1,n-1)$ or $(b_x-n,b_y)$, she can enlarge Bob's forbidden region to a square (regions $S_n(b_x-n,b_y)$ and $S_n(b_x-n,-1)$, respectively). Bob can only occupy one of the two positions in one move. Hence, regardless of Bob's choice, Alice will be able to secure a region large enough to construct a copy of $R(n,n)$.

%
%

%
Among the four moves, at most two will be outside Bob's forbidden region. Hence, Alice will win using at most $n^2+2$ moves.


\item[Three cells occupied] Now, suppose Alice is able to occupy $(2n,0)$ with her third move, but then Bob 
occupies location $(3n,0)$. As always, Bob's first move must be in the halfplane $x<0$. We also know that another one must be in the halfplane $x>2n$. Hence, Bob can have played at most once in the rectangular region $\mathcal{R}=[0,2n]\times[-n,n]$. Since Alice occupies positions $(0,0)$, $(n,0)$, and $(2n,0)$, Bob's move in $\mathcal{R}$ (if any) must either be strictly above the halfplane $y>0$ (or strictly below). Without loss of generality, we assume that the region $[0,2n]\times [0,n]$ is empty of Bob's moves. In this case, Alice's third move will be $(n,n-1)$ (see Figure~\ref{Sn-case3-1}).
\begin{figure}[h]
\center
\includegraphics[width=0.6\textwidth]{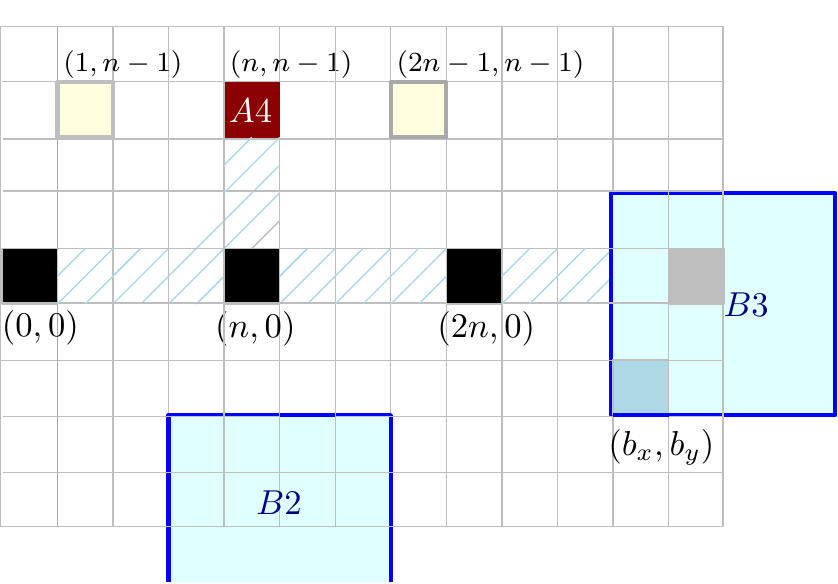}
\caption{If Bob forbids Alice from occupying $(3n,0)$ with his third move $B3$, 
then Alice plays her fourth move in the region above or below the horizontal strip 
that is free of Bob's animals: in this example at $A_4$ is above this strip.
The two cells $(1,n-1)$ and $(2n-1,n-1)$ are cells one of which Alice will try to occupy next in order to win.}
\label{Sn-case3-1}
\end{figure}
Similar to the previous case, Alice can prevent Bob from playing inside an $n\times n$ region by occupying either position $(1,n-1)$ or $(2n-1,n-1)$. Since Bob's fourth move can only block one of the two positions, Alice can play in the other one and construct a copy of $R(n,n)$. 

\item[Four cells occupied] Finally, assume Alice manages to occupy the four locations $(0,0)$, $(n,0)$, $(2n,0)$, $(3n,0)$. Similar to the case in which three cells were occupied, consider the rectangular region $R'=[0,3n]\times [-n,n]$. 

After four turns, Bob can place at most three blocks in $R'$ (recall that Bob's first move is at a location to the left of the vertical line at $(0,0)$). Moreover, Bob cannot occupy any cell of the horizontal strip between $(0,0)$ and $(3n,0)$. Hence, either the region above or below the horizontal strip will contain at most one of Bob's animals. Without loss of generality, we can assume that the upper half has none or one of Bob's pieces.

If this half is empty of Bob's animals, Alice will proceed as in the case where three cells are occupied. Her fifth move will be to play at position $(n,n)$ and afterwards $(0,n)$ or $(2n,n)$ depending on Bob's move. In either of the two cases, $n^2+3$ moves will be sufficient to construct a copy of the square.

\begin{figure}[h]
\center
\includegraphics[width=0.6\textwidth]{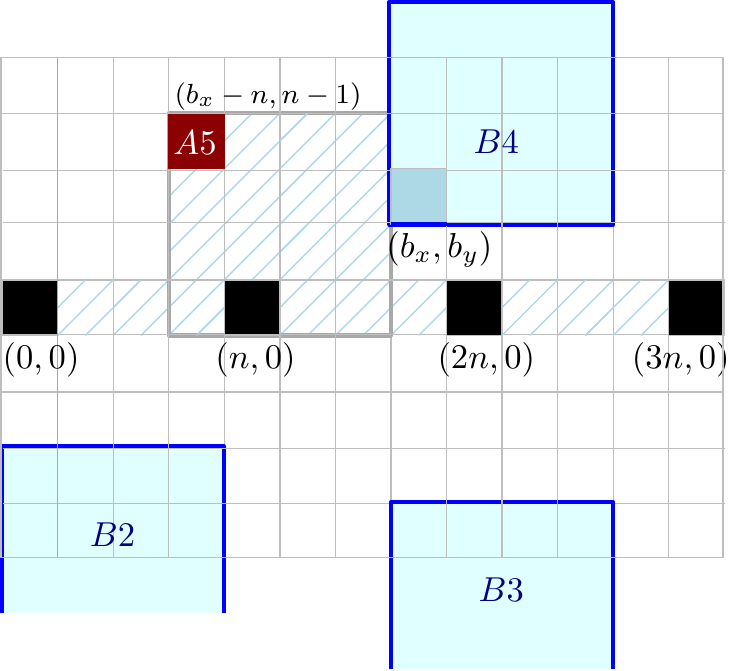}
\caption{If Alice plays the four locations $(0,0)$, $(n,0)$, $(2n,0)$, and $(3n,0)$ then with her fifth move $A5$ 
she will play inside the region above or below the horizontal strip
that has at most one of Bob's animals: in this example at $A_5$ is above this strip.
The striped rectangular area is a region that Bob can no longer occupy.}
\label{Sn-case4-1}
\end{figure}
Finally, it remains to consider the case in which Bob has placed a single animal $S_n(b_x,b_y)$ for some $0 \leq b_x \leq 3n$ and $0 < b_y \leq n$. Additionally, we assume that $b_x\geq n$ (if necessary, we can flip the board vertically to obtain this property, see Figure~\ref{Sn-case4-1}). In this case, Alice plays her fifth move at position $(b_x-n, n-1)$. This move will prevent Bob from occupying any position of the region $S_n(b_x-n,0)$. Moreover, only three out of the five moves of Alice have been placed  outside $S_n(b_x,b_y)$, hence Alice wins again in $n^2+3$ moves.
\end{description}
\end{proof}

Observe that $n^2$ is a trivial lower bound on the number of moves of any winning strategy or $R(n,n)$. Our strategy only uses at most 3 additional moves, which leads us to believe that our strategy is optimal (in the sense that no other strategy can construct $R(n,n)$ with fewer moves).

\section{Cannibal animals (losers) and pairing strategy}

In this section we demonstrate
several strategies for Bob that prevents Alice from winning. 
By Observation \ref{obs_three}, The game becomes more interesting when the animal has 5 or more cells, since we will show that there exist both winning and losing polyominoes. 

We start by using the well-known concept of \emph{pairing strategy}. We note that this strategy has been successfully used in many other 
combinatorial games~\cite{tttsurvey}. 
We start with a simple strategy for Bob that works for the $O(n,m,k)$ animal: 

\begin{figure} 
\center
\includegraphics[width=0.6 \textwidth]{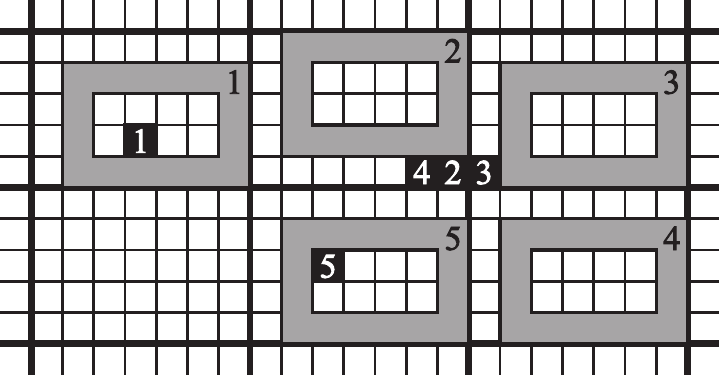}
\caption{Winning strategy for Bob for $O(n,m,k)$  
(in this example, for $O(4,6,1)$).
Alice's moves are marked in black and Bob's in gray. 
The numbers on the cells represent the order in which the cells are occupied.
Since the block inside which Alice's $4$th move is played
already includes Bob's animal, Bob's $4$th move is 
played in another arbitrary block.
}
\label{fig_O}
\end{figure}

\begin{theorem}\label{th_ring} 
$O(n,m,k)$ is a cannibal for any $n,m\geq 3$ and $k < \min \{ n/2, m/2 \}$.
\end{theorem} 
\begin{proof}
Bob virtually partitions the playing-board into blocks of size $(n+k)\times (m+k)$. 
That is, we define the block $B_{ij}$ as the rectangle $[i(n+k),(i+1)(n+k)-1]\times [j(m+k), (j+1)(m+k)-1]$ (as shown in Figure~\ref{fig_O}). The strategy for Bob is to place his animal inside the block where Alice played her last move. 
After Alice plays, Bob checks which block her last move belongs to; 
if he has already played an animal in the same block, he simply plays in an arbitrary empty block 
(e.g., Bob's $4$th move in Figure~\ref{fig_O}). 
Note that since the playing board is infinite, Bob can always play these moves. 
Further note that, with this strategy any rectangular region free of Bob's pieces has either height or width at most $2k$. Since $k < \min \{ n/2, m/2 \}$, Alice will never be capable of constructing a copy of $O(n,m,k)$.
\end{proof}

This pairing strategy can also be applied to other animals, such as the $L(n)$. 

\begin{theorem}\label{th_l} 
For any $n \geq 2$, the $L(n)$ animal is cannibal. 
\end{theorem}
\begin{proof}
The proof is analogous to the proof of Theorem~\ref{th_ring}. This time we partition the board into blocks of size $2n\times 2$. Observe that, if Alice plays in an empty block, Bob also can place a copy of $L(n)$ in the block (with the appropriate reflection).
 With this strategy, it is easy to see that Alice will not be able to create a connected polyomino of size five or larger. In particular, she will not be capable of constructing any copy of $L(n)$ (other than $L(1)$). Recall that by Observation \ref{obs_three}, the $L(1)$ is cannibal. Hence, no pairing strategy can work for $L(1)$.   
\end{proof}

In Section \ref{sec_win} we showed that squares are non-cannibals. Surprisingly, the removal of a single interior cell from a square animal can transform it into a cannibal.

\begin{lemma}\label{lem_squa2} 
For any integer $n \geq 4$, let $A$ be the $R(n,n)$ animal in which a single interior cell whose distance to the boundary 
is at least $\lfloor n/4\rfloor$ units has been removed.  Then $A$ is a cannibal.  
\end{lemma} 
\begin{proof}
The proof of this claim also uses the pairing strategy, where this time we partition the board into blocks of size $(n + \lfloor (n-1)/2 \rfloor)\times (n + \lfloor (n-1)/2 \rfloor)$. 
It is easy to see that if the removed cell is at least $\lfloor n/4\rfloor$ units away from the boundary, 
then Bob can always play his animal inside the same block as Alice's last move. 

Assume that Alice is able to construct a copy of the animal on the board. Observe that this animal can intersect with at most 4 blocks. By the pigeonhole principle, there would be a block in which Alice's pieces form a square 
of size at least $\lceil n/2 \rceil \times \lceil n/2 \rceil$ (possibly with interior cells removed). 
However, this cannot occur since Bob also occupies the same block with an $n\times n$ square.
\end{proof}

We note that we have been unable to use a similar pairing strategy when the hole is close to the boundary. In all the partitioning strategies we considered, Alice was able to create a copy of the polyomino. This pairing strategy works for many types of polyominoes. However, in some cases we might need a more careful partitioning of the grid into blocks:

\begin{figure} 
\center
\includegraphics[width=0.60 \textwidth]{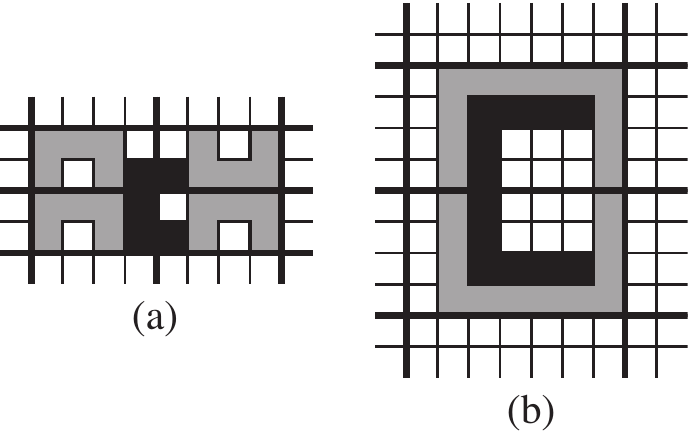}
\caption{Examples of failed partitions.}
\label{fig_U-fail}
\end{figure}
\begin{figure} 
\center
\includegraphics[width=0.60 \textwidth]{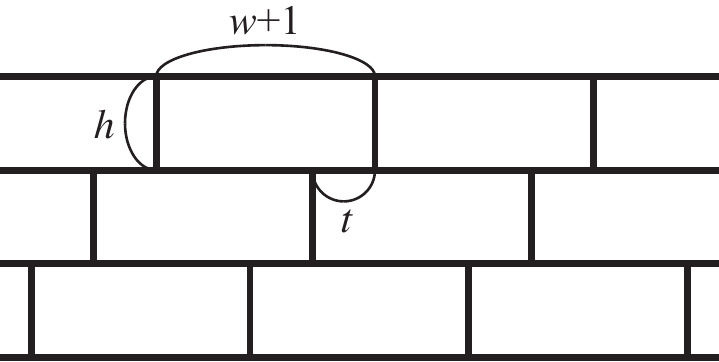}
\caption{Tiling and shift size $t$.}
\label{fig_shift}
\end{figure}
%
\begin{theorem}\label{th_u} 
For any $h,w\in \mathbb{N}$ (other than $(h,w)=(2,4)$), the $U(h,w,1)$ animal is cannibal. 
\end{theorem}
\begin{proof}
Bob virtually partitions the playing board into blocks of size $(w+k)\times h$. 
But if he arranges these blocks naively, 
there might be ``cracks'' between Bob's animals in which Alice could construct her animal
(see Figure~\ref{fig_U-fail}). 
To avoid such cracks, Bob must slant his partition, 
thus tiling the grid with blocks with a shift of size (distance) $t$ (Figure~\ref{fig_shift}). 
We define the block $B_{i,j}$ as the rectangle $[i(w+1)+jt,i(w+1)+w+jt]\times [jh, jh+h-1]$. 
%
The exact value of the slant depends on the parameters $w$ and $h$:  
\begin{description}
\item[$h=2$ (and $w \neq 4$):] $t=2$. 
\item[$h \geq 3$ and $2h-2 \geq w \geq h-2$:] $t= \lfloor (w+1)/2 \rfloor$.
\item[Otherwise:] No slant is necessary (i.e., $t=0$). 
\end{description}

It is easy to show that with such a partition, Alice will be unable to construct her animal. 
\end{proof}

By combining Theorems \ref{th_l} and \ref{th_u} we can prove the existence of 
cannibal animals of any size. For example, the polyomino $U(2,n-2,1)$ is a cannibal animal of size $n$ for any $n\geq 5$ (except for $n=6$). If $n=6$, an example of a cannibal animal would be $L(2)$. The above result combined with Theorem~\ref{theo_rect} this allows us to show the existence of both cannibal and non-cannibal animals of any size.

\begin{cor}\label{cor_5more} 
For any $n\geq 5$, there exists a cannibal and a non-cannibal polyomino of size $n$.
\end{cor}


We now introduce another idea to generate new cannibal animals from known cannibal animals. 
Let $A$ be an animal and let $C$ be a subset of cells of $A$. 
Then $A\setminus C$ is an animal created  by removing $C$ from $A$ (this operation will only be considered when $A\setminus C$ is connected). We say that $C$ is an {\em outer piece} 
if we can locate a second disjoint copy of $A\setminus C$ that overlays with 
a part of 
the removed piece $C$ of the first copy (even if partially); we call $C$ an {\em inner piece} otherwise. See Figure~\ref{fig_inner}. 

Notice that even if $C$ and $C'$ are both inner pieces, 
$C \cup C'$ need not be so. However, the superset of an outer piece must be an outer piece. 

\begin{figure} 
\center
\includegraphics[width=0.7\textwidth]{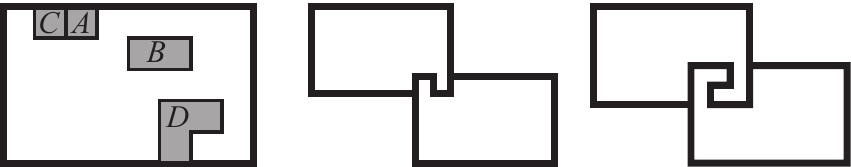}
\caption{$A$ and $B$ are inner pieces.  
$C$ and $D$ are outer pieces since a second copy covers 
a part of the piece as seen in the right examples.}
\label{fig_inner}
\end{figure}

\begin{lemma}[Punching Lemma]\label{lem:punch}
Let $A$ be a cannibal and let $C$ be an inner piece of $A$. The animal $A\setminus C$ is also a cannibal. 
\end{lemma}

\begin{proof}
Assume otherwise that $A\setminus C$ is non-cannibal; By definition, Alice is capable of constructing a copy of $A\setminus C$ without Bob preventing it. Consider now the removed piece $C$ of the animal Alice constructed. Because $C$ is an inner piece, this position cannot be occupied by Bob. Hence, Alice can afterwards occupy this position in subsequent rounds to form animal $A$. Thus, we obtain a contradiction.
\end{proof}

Note that the reciprocal is not always true (see for example Lemma \ref{lem_squa2} and Theorem \ref{theo_rect}). As a simple application of this lemma, we have the following result:

\begin{theorem}
For any integer $n \geq 4$, 
let $S'$ be an animal $R(n,n)$ in which any number of interior cells have been removed. 
If at least one of the removed cells has distance $\lfloor n/4 \rfloor$ or more to the boundary,  
then
$S'$ is a cannibal. 
\end{theorem}




\section{Concluding remarks}

In Harary's generalized tic-tac-toe, some monotone properties hold; 
these properties include ``increasing the size of the board helps Alice'' 
and ``increasing the animal helps Bob.''
However, such properties do not hold for the cannibal animal game, making it deeper and more interesting. 
We also note that the cannibal property of many other animals is still left unsolved. 
Among them is the $U(2,4,1)$ animal, which we conjecture to be a cannibal. 
  We conjecture that all $4$-cell-animals are also non-cannibals, and consequently, the $5$-cell-animal $U(2,3,1)$ would be the smallest cannibal. Another problem that remains open is what happens with the squares $R(n,n)$ in which one or more interior cells have been removed, and the distance of these removed cells to the boundary is less than $\lfloor n/4 \rfloor$ units away from the boundary.
  
Finally, we conclude with an open problem posed by an anonymous referee; observe that the only arbitrarily large non-cannibal animals that we know of are rectangles. So, it would be interesting to know if there exist arbitrarily large non-cannibal animals (other than rectangles).

\section*{Acknowledgments}
\small{
We are deeply grateful to Professor Ferran Hurtado for his valuable comments during discussions on early stages of this paper. 
}



\small 
\bibliographystyle{abbrv}

\end{document}